%% file: ASILOMAR2023.tex
\documentclass[conference]{IEEEtran}
\IEEEoverridecommandlockouts

\usepackage[ruled]{algorithm2e}
\usepackage{color,soul}

\SetAlFnt{\small}
\SetAlCapFnt{\small}
\SetAlCapNameFnt{\small}
\SetAlCapHSkip{0pt}
\IncMargin{-\parindent}

\usepackage{algorithmic}
\usepackage{amsfonts}
\usepackage{amsmath}
\usepackage{amsthm}

\newtheorem{lemma}{Lemma}

\newtheorem{definition}{Definition}

\newtheorem{assumption}{Assumption}

\usepackage{nth}

\ifCLASSINFOpdf
  \usepackage[pdftex]{graphicx}
\usepackage{epstopdf}
\else
  \usepackage[dvips]{graphicx}
\fi

\usepackage{ragged2e}
\usepackage[caption = false]{subfig}

\usepackage{url}

\def\BibTeX{{\rm B\kern-.05em{\sc i\kern-.025em b}\kern-.08em
    T\kern-.1667em\lower.7ex\hbox{E}\kern-.125emX}}

\begin{document}

\title{SCC5G: A PQC-based Architecture for Highly Secure Critical Communication over Cellular Network in Zero-Trust Environment  \thanks{This material is based upon the work supported by the US National Science Foundation under Grant No. IIP-2204502 and the Air Force Office of Scientific Research under award number FA9550-20-1-0090.}}

\author{
	\IEEEauthorblockN{
	Mohammed Gharib\IEEEauthorrefmark{1},
	Fatemeh Afghah\IEEEauthorrefmark{1} 
	}
	\IEEEauthorblockA{\IEEEauthorrefmark{1}Department of Electrical and Computer Engineering, Clemson University, Clemson, SC,  USA\\ E-mail:\{alghari,fafghah\}@clemson.edu}

}


\maketitle

\begin{abstract}



5G made a significant jump in cellular network security by offering enhanced subscriber identity protection and a user-network mutual authentication implementation. However, it still does not fully follow the zero-trust (ZT) requirements, as users need to trust the network, 5G network is not necessarily authenticated in each communication instance, and there is no mutual authentication between end users. When critical communications need to use commercial networks, but the environment is ZT, specific security architecture is needed to provide security services that do not rely on any 5G network trusted authority.  In this paper, we propose SCC5G Secure Critical-mission  Communication over a 5G network in ZT setting. SCC5G is a post-quantum cryptography (PQC) security solution that loads an embedded hardware root of authentication (HRA), such as physically unclonable functions (PUF), into the users' devices, to achieve tamper-resistant and unclonability features for authentication and key agreement. We evaluate the performance of the proposed architecture through an exhaustive simulation of a 5G network in an ns-3 network simulator. Results verify the scalability and efficiency of SCC5G by showing that it poses only a few kilobytes of traffic overhead and adds only an order of $O(0.1)$ second of latency under the normal traffic load. 
\end{abstract}


\input{introduction}

\input{relatedWork}

\input{proposedMethod}

\input{evaluation}

\input{conclusion}

\nocite{*}
\bibliographystyle{IEEEtran}
\bibliography{References}


\end{document}

%% file: introduction.tex
\section{Introduction}
Communication is an essential component of any critical-mission application such as disaster relief, military missions, etc. In many cases, mission-specific communication might not be available due to the coverage shortage, capacity limitation, and very fast network topology change. Availability, network capacity, mobility support,  and large coverage area of cellular networks make them an exceptional choice for critical-mission use when a dedicated communication link is unavailable. Three different scenarios are discussed in \cite{anshu}  where the mission-critical applications require 5G communication, for machine-to-machine communication, for hypersonic weapon command and control, and for millimeter wave (mmWave) usecases where there will be a large bandwidth with ultra-low latency. The focus of the Department of Defence (DoD) on 5G testing and experimentation for warfighting is also evidence of the importance of mission-critical 5G communication. DoD invested more than 600 million in 2020 for 5G testbeds \cite{dodCellular}, deployed at 12 sites across the country.  Working with more than 100 industry partners, DoD expects the full operation of the dual-use 5G cellular network in 2023 \cite{dodCellular}. Besides the DoD effort, many researchers proposed various mission-critical use cases for 5G networks and showed how they can advance such mission operations \cite{lukman,nato,gronsund,liao,elmasri}.

5G network, in addition to the aforementioned cellular network advantages, has a striking enhancement in terms of communication security in comparison with its predecessors, especially in identity protection and utilizing more robust encryption algorithms. However, security is still the main concern in utilizing them for critical-mission communication since 5G does not fully follow the National Institute of Standard and Technology (NIST) and the Department of Defense (DoD) recommendation of considering ZT architecture. In ZT architecture there is no pre-assumed trust and every communication instance should be protected by verifying the communicating parties regardless of their location and ownership \cite{nistZT,dodZT}.  
Critical-mission agents not only deal with a ZT network where they cannot rely on the network's regular security protocols, but they are also at constant risk of physical attacks when closely interacting with adversary agents in ``unknown'' or ``contested'' environments. 

Generally, three main reasons make a significant need for a new architecture for critical-mission communication over a 5G network, $i)$ in a 5G network, the mutual authentication is between the user and the network, not between the end users, $ii)$ the entire authentication process is based on the Sim card which can be stolen, and $iii)$ the network is responsible for the authentication verification, and hence the critical-mission agents need to trust the network. Regardless of the application domain, NIST \cite{PQC-NIST} and NSA \cite{pqcnsa} recommended that critical-mission communication should be migrated to post-quantum cryptography (PQC) architectures in addition to accommodating ZT requirements. To achieve this goal, a robust PQC-based key management system capable of providing end-to-end encrypted communication is required where each node should be authenticated for each communication instance.

In this paper, we present a PQC-based security architecture for mission-critical communication, utilizing a 5G cellular network as the communicating layer, considering ZT settings. This architecture is referred to as \textit{SCC5G, Secure Critical-mission Communication over a 5G network}. It protects the agents' devices, e.g. cell phones, ground vehicles, unmanned aerial vehicles, IoT devices, etc., from revealing secret information even if they are targeted by physical hijacking attacks. SCC5G forces a mutual authentication for each communication session to fulfill the ZT requirement. To protect the critical-mission users' devices, SCC5G provides each user with a hardware root of authentication (HRA), which is a pre-loaded hardware in the user device. The HRA should be able to take a challenge and respond with a unique  response. SCC5G requires the HRA to be physically unclonable, tamper-resistant, and have high entropy in its challenge-response pairs, i.e. have a large enough state-space of challenge-response pairs. Commercial 5G networks can provide mutual authentication among users using the users' SIM cards as the hardware root of trust (HRT). However, a SIM card is just used as hardware identification and could be cloned or tampered. 

We assume that the commercial 5G network cooperates with the critical-mission agency to add trusted authority referredd to as HRA authentication and keying function (HAKF) into its structure. SCC5G keeps the security protocols of the 5G network and adds another layer of security by utilizing the HRA devices embedded in the users' communication devices, to provide an asymmetric cryptosystem that does not rely on any 5G network trusted authority. The idea behind SCC5G is to bind the public-private key pair generated by a critical-mission user's device to the embedded HRA by using the hashed readings stored in the HAKF. It will verify the originality of the generated public key as well as the identity of the device. The entire solution relies on ID verification in every communication session to address the ZT requirements. The unclonability and high entropy properties of the HRA guarantee the uniqueness of the verification part. While any hardware with the mentioned properties could be utilized within SCC5G architecture, we suggest using a physical unclonable function (PUF) as an HRA. PUFs exploit the fabrication variation of the microelectronic devices to generate a unique response for the corresponding input. 

SCC5G has an initialization phase before the network starts its operation. At the initialization phase, the network admin chooses the required security parameters and forms a certification storage. The certification storage is used for certificate verification and includes non-confidential information. It includes the hashed values of the image of each HRA, bound to the ID of the carrying user. The network starts its operation after the initialization phase. By starting the operation, each user can utilize the cellular network to connect to any other  user. The proposed architecture encrypts the
communication from end to end and authenticates end users for any communication instance. In SCC5G, users are responsible for generating their cryptographic keys based on the PQC cryptosystem and the pre-defined cryptographic parameters. The security and the performance of the proposed architecture are comprehensively studied. We show that SCC5G carries the NIST-recommended properties \cite{nistKeyManegmet} for the secure key management system and it is resistant to well-known key management attacks such as man-in-the-middle (MitM).  We use network simulator ns-3  \cite{ns3} to evaluate the performance of the proposed architecture and verify its efficiency and scalability. Results show that both the communication overhead and latency of SCC5G  architecture are negligible with the order of $O(1)$ KB and $O(0.1)$ second, respectively, in the network with a normal load. 


The main contribution of this paper is to propose a novel architecture for utilizing the existing 5G cellular network for critical-mission communication in a zero-trust setting. It adds another layer of security by encrypting the data communication in an end-to-end manner, mutually authenticating the end users, and supporting PQC cryptography. It further enhances data transfer communication by providing a session key agreement for data communication. Comprehensive security and performance evaluation  validates the feasibility and dependability of the proposed architecture.       

The rest of this paper is organized as follows. We review the related work in Section (\ref{sec::relatedWork}) and present the system model in Section (\ref{sec:systemModel}). The details of the proposed architecture are reviewed in Section (\ref{sec::proposed}). We evaluate the security and the performance of the proposed architecture in Section (\ref{sec::evaluation}).
Finally, we conclude the paper and mention the future directions in Section (\ref{sec::conclusion}).


%% file: relatedWork.tex
\section{Related Work}
\label{sec::relatedWork}

There are over 40 types of suggested PUF technologies where some of which are based on the embedded memory of the hardware. The advantage of utilizing embedded memory as a PUF is the fact that they are widely available in cyber-physical systems. Examples of the utilized embedded memories include SRAMs, DRAMs, Flash RAMs, MRAM, and ReRAM \cite{SHAMSOSHOARAsurvey}. The ReRAM advantage for SCC5G over the other memory-based PUF technologies is its tamper-resistant property and self-destruction mechanism in case of outsider manipulation \cite{reram}. However, there are two main challenges in utilizing PUF in cryptosystems, the need for a trusted server to store the image of challenge-response pairs, and the extreme sensitivity of the entire system to the PUF response variation. 

The PUF-based security solutions have been mostly utilized in trustable networks since they require a secure server to keep a record of original PUF responses. The initial PUF responses are generated with a set of challenges and are stored as references in secure servers. During the authentication cycle, PUFs are challenged to generate fresh responses that should match the reference responses stored in the network. 
Ren et al. \cite{ren} proposed a PUF-based group authentication scheme, where a repository of all challenge-response pairs is stored in the Authentication Server Function (AUSF) in the home network. As an instance, Alladi et al. \cite{5gDrone} proposed an authentication scheme for drone-assisted 5G networks. While the protocol offers user-to-user mutual authentication, it stores the entire challenge-response pairs in the connected base station and assumes that all base stations have secure communication in between. Mall et al. \cite{iotPufSurvey} published a survey including 44 PUF-based authentication and key agreement protocols for the Internet of thing (IoT), wireless sensor network (WSN), and smart grids. However, all of the reviewed works need to store the plain challenge-response pairs somewhere in the network. Although the mentioned requirement may not seem to be unacceptable in WSN, IoT, and smart grids, it is an obvious security threat to critical-mission communication networks. SCC5G architecture does not store the challenge-response pairs, instead, it just keeps the hash value of the PUF image in HAKF. Although the HAKF is not a 5G component and could be trusted by the critical-mission agents, it does not have access to the PUF images.    

The use of PUFs for cryptographic key generation is the most challenging problem since even a single-bit mismatch between the challenge and response could not be tolerated. Robust secret key generation heavily relies on fuzzy extractors\cite{PUF_PAINE} and error-correcting methods which use a set of publicly available helper data. Most of the current PUF technologies suffer from vulnerability to side-channel analysis. The newly developed ultra-low power tamper-resistant ReRAM-based PUFs with a self-destruction mechanism, however, represented promising performance in 
 mitigating this shortcoming \cite{PUF_PAINE,reram}. The developed masking techniques added  also another layer of protection for the PUF-based key generation against helper data manipulation attacks \cite{Amir_ACCESS_2022}. Noting the extreme sensitivity of PUF-based key generation to variations of PUF responses, such methods require long helper data and trustable servers, which is not feasible in a zero-trust environment. In this work, we utilize the embedded PUFs as a means to certify the devices' public keys.

%% file: proposedMethod.tex
\section{System Model}
\label{sec:systemModel}

SCC5G utilizes the existing 5G cellular communication as an underlay and adds an overlay secure communication level at the top of it. The sage cloud in Fig. (\ref{fig:framework}) represents the existing 5G cellular network where both third-generation partnership project (3GPP) users, i.e. SIM card-based users, and non-3GPP users, i.e. users without SIM card chipset on their devices, can connect to the network. Similarly, the cloud with cantaloupe color represents the overlay formed by SCC5G architecture for critical-mission secure communication, where the users can be 3GPP or non-3GPP users, but need to have their own HRA. We will use \textit{user} and \textit{critical-mission user} interchangeably throughout the paper, since the architecture is proposed for mission-critical users, even in the presence of the typical 5G users. All non-3GPP users, critical-mission or typical, connect to the 5G core network through Non-3GPP Interworking Function (N3IWF) across IP security (IPsec) tunnel. N3IWF plays the role of a virtual private network (VPN) server allowing the non-3GPP users to establish an IPsec tunnel into the 5G core network.
\begin{figure}
\centering
    \includegraphics[width=\linewidth]{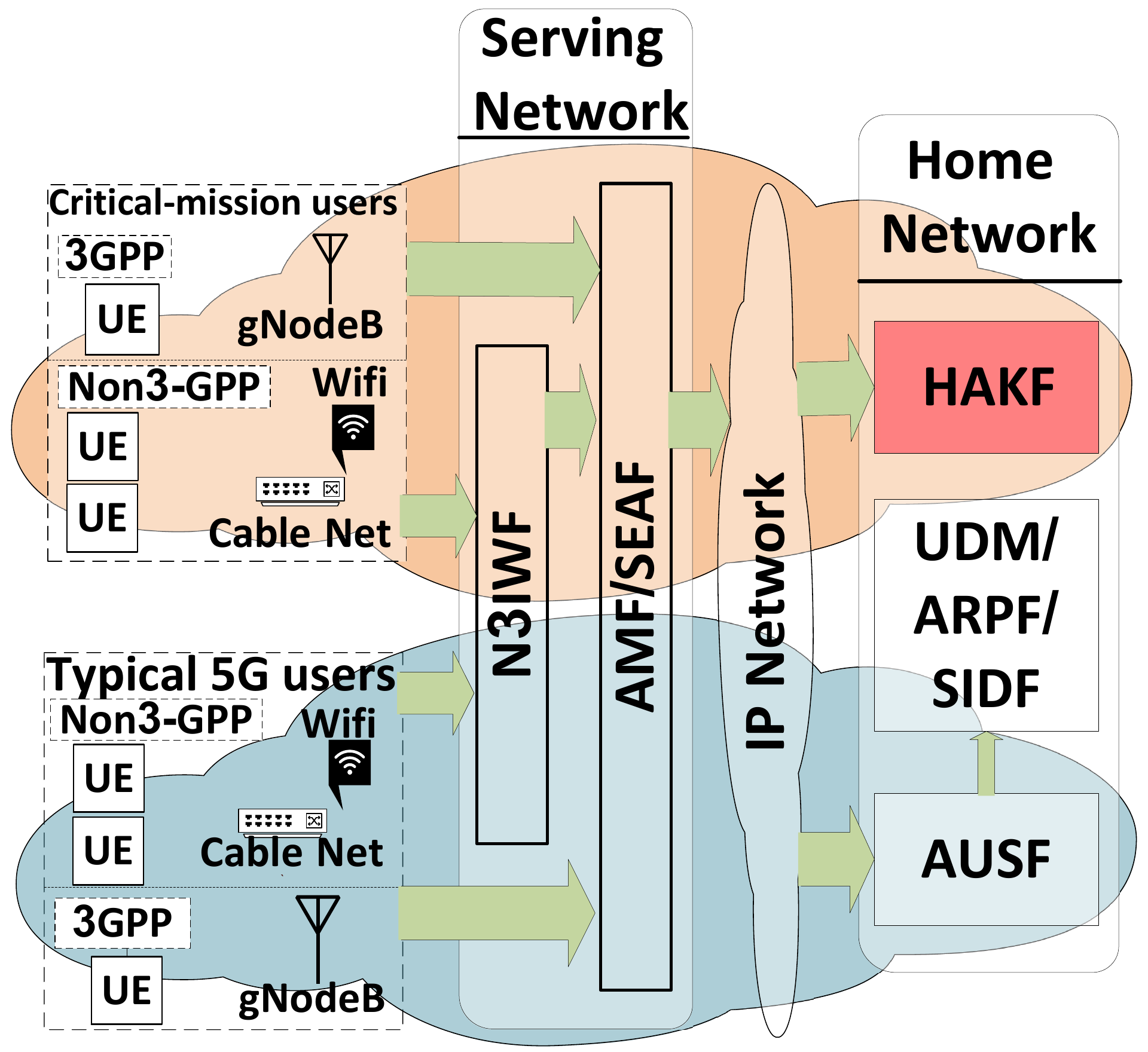}
\caption{SCC5G architecture; Secure Critical-mission Communication architecture over a 5G network.}
\vspace{-0.6 cm}
\label{fig:framework}
\end{figure}

The user's physical connections to the 5G core network or N3IWF could be held through the radio channel provided by WiFi or next generation node B (gNodeB) or even through the cable network. We divided the 5G core network into two parts based on their physical connections with the other network parts to be direct radio access or to be behind the IP network. We name the part of the core network with direct radio access as the \textit{serving network} and the latter one as the \textit{home network}. Such naming is not novel as it could be seen in \cite{cablelab}. Access and Mobility Management Function (AMF) which is responsible for connecting the users through different interfaces with 5G core networks, and Security Anchor Function (SEAF) which is responsible for initializing the 5G authentication process are located in the serving network.

The second part of the 5G core network is located behind the IP network, hence, theoretically, its building blocks could be located at any physical location. Unified Data Management (UDM) which is responsible for managing the subscribers' data, and Subscription Identifier De-concealing Function (SIDF) which is responsible for decrypting subscribers' IDs, i.e. Subscription Concealed Identifier (SUCI), are two of the 5G home network building blocks. The Authentication Server Function (AUSF) and Authentication Credential Repository and Processing Function (ARPF) are two other building blocks of the 5G home network. ARPF is responsible for authentication method selection based on the configured policy and subscriber identity. AUSF is the responsible party for making the decision about authentication or denying the authentication of the users based on some keying material that ARPF is  responsible for computing them. In SCC5G, we add HRA authentication and keying function (HAKF) to the 5G home network considering that the 5G network is cooperative with the critical-mission administration. We leave the extension of this work to be compatible with a non-cooperative 5G network for future work. HAKF contains the database of the hashed images of the HRAs. Each image is bound to the corresponding user ID. Since the HAKF is placed in the home network, it could be located at any secure location, and it does not interfere with the regular 5G core network functionality.  

\begin{assumption}
In SCC5G architecture, we assume that the 5G network is cooperative with the critical-mission network administration.
\end{assumption}

\begin{assumption}
We assume that the critical-mission agents will not communicate or respond to any critical-mission user communication if the authentication process fails.
\end{assumption}

\begin{assumption}
    \label{assum:cryptosystem}
    We assume that the utilized cryptosystem saves the confidentiality and integrity of the encrypted data for everyone except for the ones who have the proper key. The proof of the correctness of the utilized cryptosystem is beyond the scope of this paper. 
\end{assumption}

\section{Proposed SCC5G Architecture} 
\label{sec::proposed}

In this section, we present the details of SCC5G, secure critical-mission communication over a 5G network. The aim is to utilize the existing 5G network for critical-mission communication with end-to-end encryption, considering a ZT setting, where none of the network entities is assumed to be trusted. The proposed architecture load each user with an HRA and store the hash of the HRA's image in a central database. SCC5G architecture has two phases, the initialization phase, and the network operation phase. In this architecture, each user is responsible for generating its own keys. To authenticate the keys and bind them to the users, the image of the HRA is pre-hashed and loaded in a part of the network, behind the IP network. \textit{HRA authentication and keying function (HAKF)} is the network part that we added to the general 5G architecture to be responsible for storing the hashed HRA images and performing the authentication and certification-related tasks. In this section, we describe the details of the network initialization and operation phases in separate parts.

\textbf{Network Initialization Phase:} The network initialization phase includes three main steps, providing each critical-mission user with an HRA, choosing the cryptosystem parameters, and forming the certification database. During an enrolment phase performed in a secure ground, each user is loaded with an HRA. The HRA could be any hardware with four properties, $i)$ to be physically unclonable, $ii)$ to be tamper resistant, $iii)$ returns a unique response for any unique challenge, and $iv)$ has high entropy. The high entropy means having a large enough number of challenge-response pairs. The network initialization phase further contains a selection for cryptosystem parameters. SCC5G architecture is designed to be compatible with any asymmetric cryptosystem. However, we use the recent NIST-recommended PQC cryptosystem CRYSTAL-Kyber \cite{kyber} for key encapsulation mechanism (KEM) and CRYSTAL-Dilithium \cite{dilithium} for digital signature, without the loss of generality. The network admin chooses the required cryptosystem parameters at the initialization phase of the network. In the case of CRYSTAL-Kyber and CRYSTAL-Dilithium, the network admin chooses $k,d_t, d_u, d_v,$ and $q$, representing the number of columns in matrix A, the size of $t$, $u$, and $v$ vectors, and the size of the ring of the learning with error (LWE) problem, respectively. Each user, in its turn, forms the image of its HRA, i.e. all possible challenge-response pairs, hashes the responses, and sends the resulting table to the network admin. Network admin, in his turn, stores the hashed challenge-response pairs of each user bounded to their IDs in a database, for certification purposes. Considering the ZT setting, each user for any communication instance needs to be authenticated. To authenticate the user, two verification parts are used, a signed public key by the user's private key, and a hashed response of the user's HRA where the challenge is extracted from the public key. Algorithm (\ref{alg::init}) represents the pseudocode of the initialization process. 
\vspace{-0.3 cm}

\begin{table}[t!]
\caption[]{Table of notations }
\resizebox{1\textwidth}{!}{
\begin{minipage}{\textwidth}
\begin{tabular}{ l | l  }
  $C$& The challenge\\
  $R$& The response\\ 
  $f(.)$& The HRA function transforms C into its corresponding R.\\  
  $H(.)$& Hash function.\\
  $n$& Number of critical-mission users.\\
  $U_i$& The $i$th user.\\
  $HRA_i$& The HRA loaded into the user $U_i$. \\
  $\rho$&The seed to be expanded to generate matrix $A$. \\
  $\sigma$&The seed to be expanded to generate vector $s$. \\
  $e$& Additive error.\\
  $s_i$& The private key of the user $U_i$.\\
  $Y_i=(A,t)$& The public key of the user $U_i$.\\
  $m$& The plain text\\
  $c=(u,v)$& The cipher text, i.e. the concatenation of vectors $u$ and $v$.\\
  $d_t$& The dimension of vector $t$.\\
  $d_u$& The dimension of the vector $u$.\\
  $d_v$& The dimension of the vector $v$.\\
  $k$& Number of columns in matrix $A$.\\  
  $S_{i,j}$& The session key between users $U_i$ and $U_j$.\\
  $q$& A large prime number representing the\\  & ring size of the LWE problem. 
\end{tabular}
\label{tbl::notation}
\end{minipage}}
\vspace{-0.7 cm}
\end{table}


\begin{algorithm}
\caption{SCC5G Initialization Phase}
 \label{alg::init}
\begin{algorithmic}
\STATE{$[k,d_t,d_u,d_v,q]\leftarrow$defineNetParam()}, \COMMENT{The network admin chooses the network parameters.}
\STATE{Form($CertificateDB$)}, \COMMENT{The network admin forms a database $CertificateDB$ for certification process.}
\FOR{$(i=1:n)$}
\STATE{Load($U_i,HRA_i$)}, \COMMENT{The $i$th user is loaded with an HRA.}
\STATE{$j\leftarrow 1$}, \COMMENT{A counter for the number of challenge-response pairs.}
\STATE{Form($DB_i$)}, \COMMENT{User $U_i$ forms its own database $DB_i$.}
\WHILE{$ ! endOf(HRA_i)$}
\STATE{$H_j \leftarrow $ Form($ID_i, C_j,H(f(C_j))$)}, \COMMENT{User $U_i$ forms a database entry containing its ID,  challenges $C_j$, and the hashed value of the corresponding response.}
\STATE{Concat($H_j, DB_i$)}, \COMMENT{User $U_i$ concatenate the $H_j$ entry to its own database.}
\ENDWHILE
\STATE{Concat($DB_i,CertificateDB$)}, \COMMENT{Newtork admin concatenate the user $U_i$ database into the certification database.}
\ENDFOR
\end{algorithmic}
\end{algorithm}
\vspace{-0.4 cm}


\textbf{Network Operation Phase:} In this architecture, each user is responsible for generating its pair of public-private keys. When the network starts working, i.e. operational phase of the network, each user chooses two random numbers $\rho$ and $\sigma$ to generate 
\begin{equation}
    A\sim R^{k\times k}_q:= Sam(\rho),
\end{equation} 
\begin{equation}
    (s,e)\sim \beta_{\eta}^k \times\beta_{\eta}^k:= Sam(\sigma),
\end{equation} 
and
\begin{equation}
    t:=Compress_q(As+e,d_t),
\end{equation} 
where $(A,t)$ and $s$ are the user's public key and private key, respectively. Here, $A$ is a $k \times k$ matrix in the Galois Field of order $q$ where $Sam(.)$ and $Compress(.)$ are sampling and compression functions, respectively. In this system, $A\times s +e=t$, where $e$ is an additive error and the security of the system is based on the hardness of module learning with error (M-LWE) problem. For more detailed information about how this PQC system works, we refer the researchers to read CRYSTAL-Kyber main paper \cite{kyber}. 

\begin{figure}
\centering
    \includegraphics[width=\linewidth]{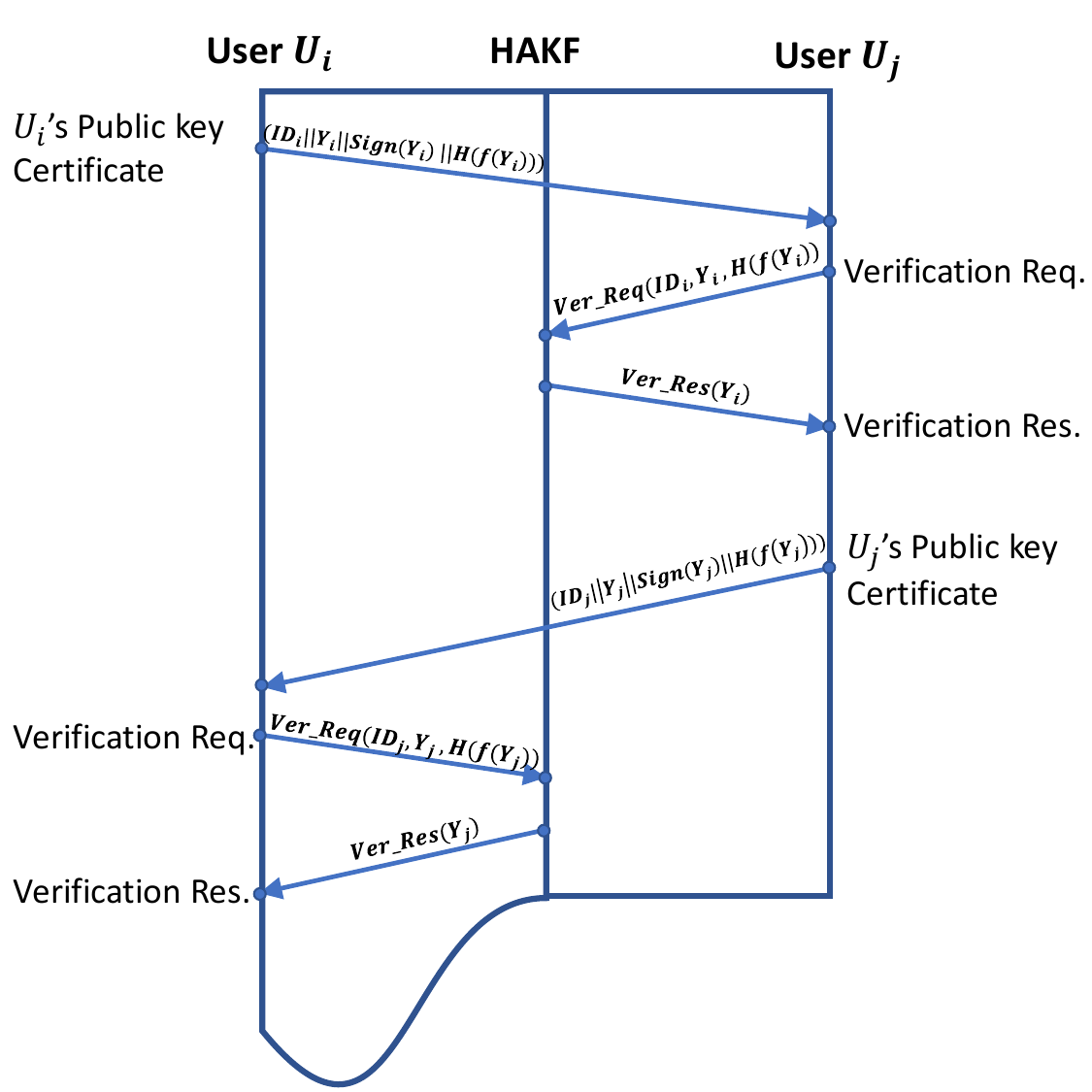}
\caption{SCC5G mutual authentication handshake process. }
\vspace{-0.8 cm}
\label{fig::handshake}
\end{figure}
\begin{algorithm}
\caption{SCC5G Handshake Process }
 \label{alg::handshake}
 \SetKwBlock{useri}{User $U_i$}{}
 \SetKwBlock{userj}{User $U_j$}{}
 \SetKwBlock{hakf}{HAKF}{}
\begin{algorithmic}
\STATE{\textbf{Handshake process between $U_i$ and $U_j$, initiated by  $U_i$.}}
\BlankLine
\nl \useri
{
\STATE{$R_i\leftarrow f(Y_i)$}
\STATE{$Sgn_i\leftarrow Sign(Y_i)_{s_i}$}, \COMMENT{Function $Sign(m)_{s_i}$ signs the message $m$ with the key $s_i$.}
\STATE {$Cert_i\leftarrow$ Concat($ID_i, Sgn(i),H(R_i)$)}, \COMMENT{The function Concat() concatenates the included arguments.}
\STATE {Send($Y_i,Cert_i)\rightarrow U_j$ }
}

\BlankLine
\nl \userj
{
\STATE{$v\leftarrow Verify(Y_i,Cert_i)$},\COMMENT{Function Verify() returns a boolean value representing the verification process result for the included sign and key.}
\IF{$(v)$}
\STATE{ver\_Req$\leftarrow$ Concat($ID_i,Y_i,H(R_i)$)}
\STATE{Send$($ver\_Req$)\rightarrow$ HAKF}
\ELSE
\STATE{Terminate the handshake process}, \COMMENT{Handshake failure}
\ENDIF
}

\BlankLine
\nl \hakf
{
\STATE{$R\leftarrow$ Retrieve($ID_i,Y_i$)},\COMMENT{Function Retrieve() retrieves the content of the provided address from the database.}
\IF{$(R==H(R_i))$}
\STATE{ver\_Res$\leftarrow$ True}
\ELSE
\STATE{ver\_Res$\leftarrow$ False}
\ENDIF
\STATE{Send(ver\_Res)$\rightarrow U_j$ }
}
\BlankLine
\nl \userj
{
\IF{($!$ ver\_Res)}
\STATE{$R_j\leftarrow f(Y_j)$}
\STATE{$Sgn_j\leftarrow Sign(Y_j)_{s_j}$}
\STATE {$Cert_j\leftarrow$ Concat($ID_j, Sgn(j),H(R_j)$)}
\STATE {Send($Y_j,Cert_j)\rightarrow U_i$ }
\STATE{$S_{i,j}\leftarrow Encapsulate(Y_i) $}
\ELSE
\STATE{Terminate the handshake process},\COMMENT{Handshake failure}
\ENDIF
}
\BlankLine
\nl \useri
{
\STATE{$v\leftarrow Verify(Y_j,Cert_j)$}
\IF{$(v)$}
\STATE{ver\_Req$\leftarrow$ Concat($ID_j,Y_j,H(R_j)$)}
\STATE{Send$($ver\_Req$)\rightarrow$ HAKF}
\ELSE
\STATE{Terminate the handshake process},\COMMENT{Handshake failure}
\ENDIF
}
\BlankLine
\nl \hakf
{
\STATE{$R\leftarrow$ Retrieve($ID_j,Y_j$)}
\IF{$(R==H(R_j))$}
\STATE{ver\_Res$\leftarrow$ True}
\ELSE
\STATE{ver\_Res$\leftarrow$ False}
\ENDIF
\STATE{Send(ver\_Res)$\rightarrow U_i$ }
}
\BlankLine
\nl \useri
{
\IF{($!$ ver\_Res)}
\STATE{$S_{i,j}\leftarrow Encapsulate(Y_j) $}
\ELSE
\STATE{Terminate the handshake process},\COMMENT{Handshake failure}
\ENDIF
}
\end{algorithmic}
\end{algorithm}


For any communication instance, a mutual authentication followed by a session key agreement between the two communicating end users is required.  Fig. (\ref{fig::handshake}) represents the handshake process in a diagram. Considering the challenge-response process of the HRA as a function $f(.)$, the user who requires to establish a new communication instance, gives its public key as the input to its HRA, takes the response $f(A,t)$, hashes the response to generate the verification part, i.e. $H(f(A,t))$, and forms a certificate by concatenating his ID, public key, signed public key by the private key, and the hashed response of the HRA. The initiating user sends the concatenated message to the destination through the 5G network. The destination receives the message and takes four actions, it checks the validity of the signed public key, sends a verification request to HAKF to verify the verification part, generates and sends a public key certificate in the same way as the source user, and calculates a session key $S_{i,j}$. The source node, in its turn, gets the certificate, extracts the public key of the destination, verifies the destination public key with the help of HAKF, and calculates a session key $S_{i,j}$. The session key can be calculated by the two endpoints when they have the public keys of one another as 
\begin{equation}
    S_{i,j}=Encapsulate(Y_i)=Encapsulate(Y_j).
\end{equation}
It is worth mentioning that the users send a certificate including their public key in all of their communications with HAKF. While the communication from each user to the HAKF includes only public information, the HAKF encrypts its response to the verification request by the public key of the requesting node, to protect the response from manipulation. The pseudocode of the SCC5G mutual authentication handshake process is represented in Algorithm (\ref{alg::handshake}). In Section (\ref{sec::evaluation}), we show that the proposed handshake has a negligible overhead, follows the NIST key management recommendations \cite{nistKeyManegmet}, and is resistant to the well-known attacks.

%% file: evaluation.tex
\section{Security and Performance Evaluation}
\label{sec::evaluation}
Security and performance of the proposed architecture are required to be carefully considered together in the evaluation process, as there may exist a highly secure architecture but carrying high complexity in terms of computation, time, storage, etc. which makes it infeasible to be implemented in the real world. In this section, we first evaluate the security of SCC5G in terms of its security properties that might be interpreted as resiliency against different well-known attacks. Then, we exhaustively evaluate the performance of the mutual authentication and session key generation algorithm which is considered the core algorithm of the proposed architecture, in terms of latency and communication overhead.

\subsection{Security Evaluation}
We show in this section that SCC5G carries the essential key management properties (explicit key authentication, key freshness, perfect forward secrecy, backward secrecy, key independence), provides the primary security services (confidentiality, integrity, authentication, authorization, and non-repudiation),  and is resistant against the most well-known key management attacks (MitM attack and ID-spoofing).

\begin{definition}
\textbf{Explicit key authentication} is the process of authenticating a key with the knowledge that the intended user owns the corresponding private key. 
\end{definition}

\begin{lemma}
SCC5G has explicit key authentication property.
\end{lemma}
\begin{proof}
Since the public key certificate in SCC5G binds the ID, public key, signed public key, and the unclonable tamper-resistant response of the user's HRA together, it authenticates the key and proves its ownership. The public key along with the sign verify that the certificate ties the public key and its private key, whereas the verification part extracted from the HRA certifies that the public key is for the intended user. Hence, SCC5G carries the explicit key authentication property.   
\end{proof}

\begin{definition}
    \textbf{Forward and backward secrecy} are two properties for the cryptographic keys which state that revealing a key should not reveal any other upcoming key or former key, respectively. In other words, there should be no relationship between a key and its successor and/or predecessor keys that make them extractable from one another. The \textit{perfect} forward secrecy ensures that there is no long-term key that if revealed may lead to revealing a new key. A long-term key example is the master private key in the cryptosystems where the entire system security is built based on a single master key. 
\end{definition}

It is obvious that in SCC5G there is no dependency between the keys, as they are randomly chosen by the users. Hence, it carries perfect forward and backward secrecy in addition to key independence property. The keys can also be updated at any time by the user, thus, SCC5G carries key freshness property, too.

\begin{lemma}
SCC5G is resistant to ID spoofing attacks. 
\end{lemma}
\begin{proof}
First, SCC5G requires mutual authentication for each and every communication instance. Second, SCC5G binds the user's ID to the key, and to the certificate. First and second together guarantee the verification of the users' IDs for each communication instance which makes the SCC5G resistant to ID spoofing attacks.  
\end{proof}

\begin{lemma}
\label{lem:mitm}
SCC5G is resistant to man-in-the-middle attacks.
\end{lemma}
\begin{proof}
In a man-in-the-middle attack, the attacker inserts itself between the user and the certifying entity or between two users and foists itself instead of the other communication party. The uniqueness of the HRA responses and the bind between the user ID, public key, and the HRA response makes the man-in-the-middle attack infeasible. 
\end{proof}

\begin{lemma}
SCC5G architecture provides confidentiality and integrity services for critical-mission users communicating over a commercial 5G network. 
\end{lemma}
\begin{proof}
To provide confidentiality, end-to-end encryption is required to assure that the encryption key(s) has/have not been revealed or forged. To provide integrity, in addition to the end-to-end encryption condition, the cryptosystem is also required to keep the encrypted data correlated such that altering or omitting any bit damages the entire block of data. SCC5G provides the users with end-to-end PQC encryption. Hence, the confidentiality and the integrity of the communication are conditioned by the confidentiality of the utilized cryptosystem (CRYSTAL-Kyber and CRYSTAL-Dilithium), and the resistance of the architecture against revealing the keys. According to Assumption (\ref{assum:cryptosystem}), the utilized cryptosystem keeps the encrypted data confidential and integrated. Since the private key of each user is generated by the user himself and not communicated over the network, the attacker cannot get access to it except by physically compromising the user. Even in the case of the physical compromise, since the key is bound to the HRA and the HRA is tamper-resistant, the attacker cannot generate a new key. The only other known way to attack the confidentiality and the integrity of communication is to perform a man-in-the-middle attack, where according to Lemma (\ref{lem:mitm}), SCC5G is resistant to this attack and hence, it provides confidentiality and integrity.  
\end{proof}

\begin{lemma}
SCC5G provides the users with authentication, authorization, and non-repudiation services.
\end{lemma}
\begin{proof}
SCC5G fulfills the ZT requirement by performing mutual authentication for each communication instance. Hence, to prove providing authentication service, we need to prove that mutual authentication works perfectly secure. The mutual authentication process, according to Algorithm (\ref{alg::handshake}), has seven steps and six data communication instances (as also shown in Fig. \ref{fig::handshake}). The first and the fourth communication instances are between the users, where they exchange their certificates. The certificates are public information transferred plainly, i.e. without encryption. The second and fifth communication instances include verifiable public  information. The rest of the communication instances are encrypted communication from HAKF to the users. Hence, to prove the security of the authentication and authorization process, we need to show that no one else can forge the certificate. Each certificate binds four values together, user ID, the user's public key, the signed public key by the corresponding private key, and the HRA verification part. Binding the user ID and the unclonable tamper-resistant HRA to the key guarantees authorization, in case of authentication success. The presence of the digital signature is also guaranteeing non-repudiation. Hence, the authentication, authorization, and non-repudiation security are now conditioned by the correctness of the certificate. To forge a certificate, the attacker needs to find the hashed value of the users' HRA response to the challenge consisting of the combination of the public key and user ID. Since the HRA is assumed to be unclonable and tamper-resistant, no one other than the corresponding user can extract the response. Hence, the certificate cannot be forged. However, if the attacker becomes able to change the HAKF response to verify a forged certificate, can perform a man-in-the-middle attack. While the communication from HAKF is encrypted, we further show in Lemma (\ref{lem:mitm}), that SCC5G is resistant to man-in-the-middle attacks. Hence SCC5G provides authentication, authorization, and non-repudiation services. 
\end{proof}

\subsection{Performance Evaluation}
To evaluate the performance of the proposed architecture, we focus on the efficiency and scalability of the proposed architecture. Since SCC5G adds a mutual authentication to each communication instance, to verify its scalability and efficiency, we need to investigate the latency and traffic overhead it poses to each communication instance. High latency and/or authentication traffic overhead lead to the infeasibility of the proposed architecture in the real world. While those parameters could be easily measured via a simple simulation scenario, the results of such an evaluation might not be dependable. Hence, we designed a comprehensive simulation scenario using ns-3 network simulator \cite{ns3} to test the feasibility and to evaluate the performance of SCC5G. We first describe the simulation scenario and then show the results of the evaluation process.

Table (\ref{tbl::simSetting}) represents the  simulation setting in brief. We simulated a 5G network with a single base station and different numbers of user equipment (UE) ranging from two to 40 in the increment steps of two. The network is a 5G cellular network working on the mid-band, including two bands one for time division duplexing (TDD) and the other for frequency division duplexing (FDD) each with 100 MHz bandwidth centered on 2 and 2.2 GHz, respectively. The TDD band includes one bandwidth part (BWP) specified for voice calls where the FDD is divided into two BWPs one for video streaming (download) and the other for gaming (upload). All users move according to a random waypoint (RWP) mobility model in a squared area cell of $500\times500$ meters. We selected these area dimensions to encompass network evaluation under congested conditions. All users generate all three types of traffic and one of them initiates an authentication process after a warmup time. We ran each simulation instance with different random generator seeds over 100 times and reported the average of the results. 
\begin{figure}
\centering
    \includegraphics[width=\linewidth]{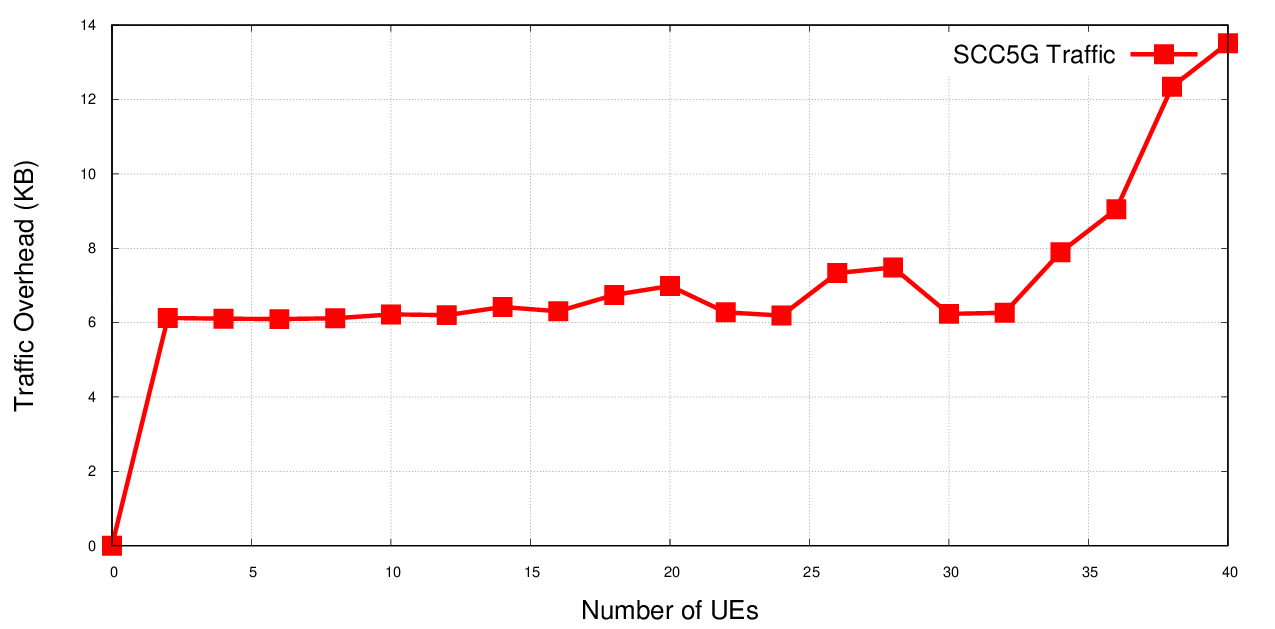}
\caption{SCC5G mutual authentication traffic overhead (KB). }
\label{fig:traffic}
\vspace{-0.5 cm}
\end{figure}

\begin{table}[t]
\caption[]{Simulation Setting }
\resizebox{1\textwidth}{!}{
\begin{minipage}{\textwidth}
\begin{tabular}{ l | l  }
 Simulator version & ns-3 3.37 with 5G-LENA module\\
  Number of UEs & $[2\quad 40]$ \\
  Area size & $500\times 500 m$\\
  Warmup time & 100 msec\\
  Speed range & $[0 \quad 20] m/s$\\
  Mobility models & RWP\\
  Traffic type & UDP \\
  Voice flow type & GBR\_CONV\_VOICE\\
  Video flow type & GBR\_CONV\_VIDEO\\
  Gaming flow type & GBR\_GAMING\\
  UE antenna  & $2\times 4$\\
  Base station antenna  & $4 \times 8$\\
  Antenna model& Isotropical \\
  Base station antenna power& 4 dB\\  
  No. of component carriers(CC) & Two $CC_0$ and $CC_1$\\
  $CC_0$ band type&TDD \\
  $CC_1$ band type&FDD\\
  No. of bandwidth parts & Three\\
  $CC_0$ Central frequency&2 GHz \\
  $CC_1$ Central frequency&2.2 GHz \\
  CC Bandwidth&100 MHz \\
  Certificate packet size& 1 KB\\
\end{tabular}
\label{tbl::simSetting}
\end{minipage}}
\vspace{-0.6 cm}
\end{table}

Fig. (\ref{fig:traffic}) shows the average amount of the imposed authentication traffic of SCC5G algorithm. This metric is directly related to the throughput of the network, as the more congested network leads to more packet drops and accordingly more retransmissions. As Fig. (\ref{fig:traffic}) shows, the increment in the network traffic slightly increases the average traffic overhead. This parameter starts with around 6 KB for the uncongested network which represents the six transmissions of one KB each for a mutual authentication process. The handshake process includes six packet transmissions where two of which are user-to-user and the rest are between the users and the HAKF. In the most congested network, this number is as low as 13 KB. Overall, this result shows that SCC5G traffic is negligible in comparison with the data transmission of such a congested network. 

Fig. (\ref{fig:latency}) shows the average time required for the mutual authentication process of SCC5G starting from the authentication initialization of the source node and ending with the authentication success at both the source and the destination nodes. The reported times are in seconds starting at about 0.25 sec for the uncongested network and slightly less than 0.55 second for the highly-congested network. While not represented here, we investigated the end-to-end delay of a single small packet latency in such a network without any extra traffic. It takes about 0.04 sec for such a packet to be transferred from the source and received by the destination node. Recalling that the mutual authentication process has six transmission instances, the overall mutual authentication time is low for the less congested network. Considering that each data transmission communication needs only one mutual authentication process to take place, the reported authentication process time sounds feasible, even for highly congested networks.   

\begin{figure}
\centering
    \includegraphics[width=\linewidth]{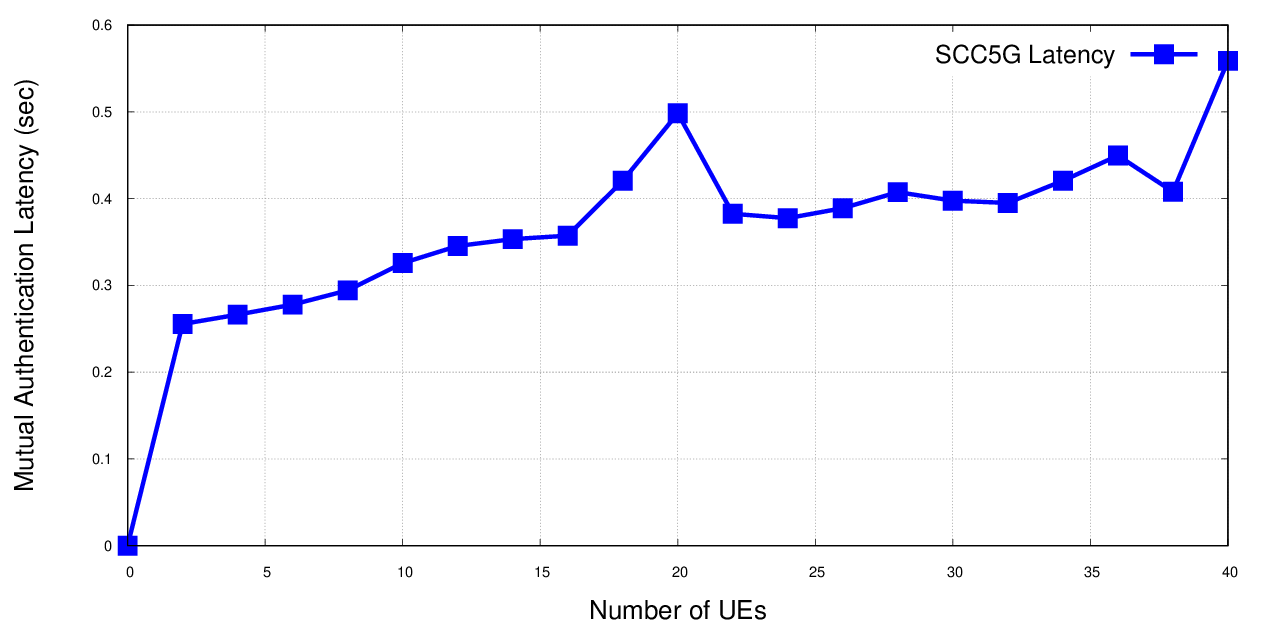}
\caption{SCC5G mutual authentication latency (sec). }
\label{fig:latency}
\vspace{-0.8 cm}
\end{figure}

%% file: conclusion.tex
\section{Conclusion}
\label{sec::conclusion}
In this paper, we present SCC5G, a secure critical-mission communication architecture over a 5G network. It includes a full architecture design with a careful protocol design for PQC end-to-end encrypted communication considering the ZT setting. The comprehensive security and performance evaluation show that SCC5G is not only feasible to be implemented, efficient, and scalable but also it is resistant to well-known security attacks such as man-in-the-middle and ID spoofing attacks. The proposed architecture is based on an assumption that the existing cellular network is cooperative with the critical-mission management system, such that the network admins are able to add a repository behind the IP network to the 5G home network. As a future direction of this work, we aim at relaxing his assumption and distributing the repository among the critical-mission users such that any $k$ out of $n$ of them can cooperate to validate a certificate.
